\documentclass{llncs}

\usepackage{times}
\usepackage{mathptm}

\usepackage{algorithmic}
\usepackage[ruled]{algorithm}
\usepackage{cite}
\usepackage{listings}
\usepackage{graphicx}
\usepackage{url}

\long\def\omit#1{\relax}    % MARK TEXT TO BE INCLUDED ONLY IN LONG VERSION

\def \figurescale {.85} % scale factor for all figures

\begin{document}

\lstset{language=Python}

\title{Optimal Embedding Into Star Metrics}

\author{David Eppstein \and Kevin A. Wortman}

\institute{Department of Computer Science \\
Universitiy of California, Irvine \\
\email{\{eppstein, kwortman\}@ics.uci.edu}
}

%\date{}

\maketitle

\begin{abstract}
We  present an $O(n^3 \log^2 n)$-time algorithm for  the following problem: given a finite metric space $X$, create a star-topology network with the points of $X$ as its leaves, such that the distances in the star are at least as large as in $X$, with minimum dilation. As part of our algorithm, we solve in the same time bound the \emph{parametric negative cycle detection problem}: given a directed graph with edge weights that are increasing linear functions of a parameter $\lambda$, find the smallest value of $\lambda$ such that the graph contains no negative-weight cycles.
\end{abstract}

\section{Introduction}
\label{section:introduction}
A \emph{metric space} is a set of sites separated by symmetric positive distances that obey the triangle inequality. If $X$ and $Y$ are metric spaces and $f:X\mapsto Y$ does not decrease the distance between any two points, the \emph{dilation} or \emph{stretch factor} of $f$ is
$$\sup_{x_1,x_2\in X}\frac{d(f(x_1),f(x_2))}{d(x_1,x_2)}.$$
We define a \emph{star metric} to be a metric space in which there exists a \emph{hub} $h$ such that, for all $x$ and $y$, $d(x,y)=d(x,h)+d(h,y)$. Given the distance matrix of an $n$-point metric space $X$, we would like to construct a function $f$ that maps $X$ into a star metric $Y$, that does not decrease distances, and that has as small a dilation as possible. In this paper we describe an algorithm that finds the optimal $f$ in time $O(n^3 \log^2 n)$.
Our problem may be seen as lying at the confluence of three major areas of algorithmic research:

\smallskip\noindent
{\bf Spanner construction.} A \emph{spanner} for a metric space $X$ is a graph $G$ with the points of $X$ as its vertices and weights (lengths) on its edges, such that path lengths in $G$ equal or exceed those in $X$; the dilation of $G$ is measured as above as the maximum ratio between path length and distance in $X$. The construction of sparse spanners with low dilation has been extensively studied~\cite{Epp-HCG-00} but most papers in this area limit themselves to bounding the dilation of the spanners they construct rather than constructing spanners of optimal dilation.  Very few optimal spanner construction problems are known to be solvable in polynomial time; indeed, some are known to be NP-complete \cite{klein_kutz} and others NP-hard \cite{1273697,edmonds_2008}. Our problem can be viewed as constructing a spanner in the form of a star (a tree with one non-leaf node) that has optimal dilation.

\smallskip\noindent
{\bf Metric embedding.} There has been a large amount of work within the algorithms community on \emph{metric embedding} problems, in which an input metric space is to be embedded into a simpler target space with minimal distortion~\cite{Lin-ICM-02}; typical target spaces for results of this type include spaces with $L_p$ norms and convex combinations of tree metrics. As with spanners, there are few results of this type in which the minimum dilation embedding can be found efficiently; instead, research has concentrated on proving bounds for the achievable dilation. Our result provides an example of a simple class of metrics, the star metrics, for which optimal embeddings may be found efficiently. As with embeddings into low-dimensional $L_p$ spaces, our technique allows an input metric with a quadratic number of distance relationships to be represented approximately using only a linear amount of information.

\smallskip\noindent
{\bf Facility location.} In many applications one is given a collection of \emph{demand points} in some space and must select one or more \emph{supply points} that maximize some objective function. For instance, the 1-median (minimize the sum of all distances from demand points to a single supply point) and 1-center (minimize the greatest distance between any destination point and a single supply point) can be applied to operational challenges such as deciding where to build a radio transmitter or railroad hub so as to maximize its utility \cite{drezner_2002}.  In a similar vein the problem discussed in this paper may be seen as selecting a single supply point to serve as the hub of a star-topology network.  In this context dilation corresponds to the worst multiplicative cost penalty imposed on travel between any pair of input points due to the requirement that all travel is routed through the hub (center) point.  Superficially, our problem differs somewhat from typical facility location problems in that the star we construct has a hub that is not given as part of the input. However, it is possible to show that the hub we find belongs to the \emph{tight span} of the input metric space~\cite{DreHubMou-DM-01}, a larger metric space that has properties similar to those of $L_\infty$ spaces. Viewing our problem as one of selecting the optimal hub point from the tight span gives it the format of a facility location problem.

Previously~\cite{Eppstein200727} we considered similar minimum dilation star problems in which the input and output were both confined to low-dimensional Euclidean spaces. As we showed, the minimum-dilation star with unrestricted hub location may be found in $O(n\log n)$ expected time in any bounded dimension, and for $d=2$ the optimal hub among the input points may be selected in expected time $O(n \, 2^{\alpha(n)} \log^2 n)$, where $\alpha(n)$ is the inverse Ackermann function. For the general metric spaces considered here, the difficulty of the problems is reversed: it is trivial to select an input point as hub in time $O(n^3)$, while our results show that an arbitrary hub may be found in time $O(n^3 \log^2 n)$.

As we discuss in Section \ref{section:linear_program}, the minimum dilation star problem can be represented as a linear program; however solving this program directly would give a running time that is a relatively high order polynomial in $n$ and in the number of bits of precision of the input matrix.  In this paper we seek a faster, purely combinatorial algorithm whose running time is strongly polynomial in~$n$.  Our approach is to first calculate the dilation $\lambda^*$ of the optimal star.  We do this by forming a \emph{$\lambda$-graph} $G(\lambda)$: a directed graph with weights in the form $w(e)=\lambda \cdot m_e + b_e$ for parameters $m_e \geq 0$ and $b_e$ determined from the input metric.  $G(\lambda)$ has the property that it contains no negative weight cycles if and only if there exists a star with dilation $\lambda$.  Next we calculate $\lambda^*$, the smallest value such that $G(\lambda^*)$ contains no negative-weight cycles, which is also the dilation of the star we will eventually create.  Finally we use $G(\lambda)$ and $\lambda^*$ to compute the lengths of the edges from the star's center to each site, and output the resulting star.

Our algorithm for computing $\lambda^*$, the smallest parameter value admitting no negative cycles in a parametrically weighted graph, warrants independent discussion.  To our knowledge no known strongly polynomial algorithm solves this problem in full generality.  Karp and Orlin \cite{ko80} gave an $O(mn)$ time algorithm for a problem in which the edge weights have the same form $w(e) = \lambda \cdot m_e + b_e$ as ours, but where each $m_e$ is restricted to the set $\{0, 1\}$.  If all $m_e=1$, the problem is equivalent to finding the minimum mean cycle in a directed graph \cite{karp78}, for which several algorithms run in $O(mn)$ time \cite{dasdan98experimental}.  In our problem, each $m_e$ may be any nonnegative real number;  it is not apparent how to adapt the algorithm of Karp and Orlin to our problem. Gusfield provided an upper bound \cite{gusfield80} on the number of breakpoints of the function describing the shortest path length between two nodes in a $\lambda$-graph, and Carstensen provided a lower bound \cite{carstensen84} for the same quantity; both bounds have the form $n^{\Theta(\log n)}$.  Hence any algorithm that constructs a piecewise linear function that fully describes path lengths for the entire range of $\lambda$ values takes at least $n^{\Theta(\log n)}$ time.   In Section \ref{section:cycle_detection} we describe our algorithm, which is based on a dynamic programming solution to the all pairs shortest paths problem.  Our algorithm maintains a compact piecewise linear function representing the shortest path length for each pair of vertices over a limited range of $\lambda$ values, and iteratively contracts the range until a unique value $\lambda^*$ can be calculated.  Thus it avoids Carstensen's lower bound by finding only the optimal $\lambda^*$, and not the other breakpoints of the path length function, allowing it to run in $O(n^3 \log^2 n)$ time.

\begin{figure}
\centering
\scalebox{\figurescale}{\includegraphics{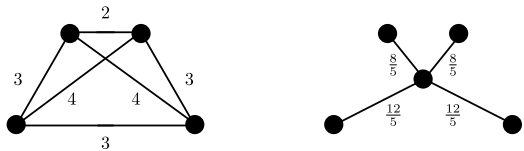}}
\caption{Example of a metric space and its optimal star, which has dilation $\lambda^*=8/5$.}
\label{figure:example}
\end{figure}

\section{Linear Programming Formulation}
\label{section:linear_program}

In this section we formally define the overall minimum dilation star problem and describe how to solve it directly using linear programming.  Our eventual algorithm never solves nor even constructs this linear program directly; however stating the underlying linear program and its related terminology will aid our later exposition.

The input to our algorithm is a finite \emph{metric space}.  Formally, a metric space $\mathcal{X}$ is a tuple $\mathcal{X}=(X,d_X)$, where $X$ is a set of sites and the function $d_X$ maps any pair of sites to the nonnegative, real distance between them.  The following \emph{metric conditions} also hold for any $x,y,z \in X$:
\begin{enumerate}
\item $d_X(x,y) = 0$ if and only if $x=y$ (positivity);
\item $d_X(x,y) = d_X(y,x)$ (symmetry); and
\item $d_X(x,y) + d_X(y,z) \geq d_X(x,z)$ (the triangle inequality).
\end{enumerate}
\noindent The input to our algorithm is a finite metric space $\mathcal{S}=(S,d_S)$; we assume that the distance $d_S(x,y)$ between any $x,y \in S$ may be reported in constant time, for instance by a lookup matrix.

A \emph{star} is a connected graph with one \emph{center} vertex.  A star contains an edge between the center and every other vertex, but no other edges.  Hence any star is a tree of depth 1, and every vertex except the center is a leaf.  Our algorithm must output a weighted star $H$ whose leaves are the elements $S$ from the input.  The edge weights in $H$ must be at least as large as the distances in $S$, and must obey reflexivity and the triangle inequality.  In other words, if $d_{H}(x,y)$ is the length of a shortest path from $x$ to $y$ in $H$, then $d_{H}(x,y) \geq d_{S}(x,y)$, $d_{H}(x,y) = d_{H}(y,x)$, and $d_{H}(x,y) + d_{H}(y,z) \geq d_{H}(x,z)$ for any vertices $x,y,z$ in $H$.

We also ensure that the \emph{dilation} of $H$ is minimized.  For any two vertices $u,v$ in some weighted graph $G$ whose vertices are points in a metric space, the dilation between $u$ and $v$ is
\[ \delta_G(u,v) = \frac{d_G(u,v)}{d_S(u,v)} .\]
\noindent The dilation of the entire graph $G$ is the largest dilation between any two vertices, i.e.
\[ \Delta_G = \max_{u,v \in G} \delta_G(u,v) . \]
Our output graph $H$ is a star; hence every path between two leaves has two edges, so if we apply the definition of dilation to $H$, we obtain
\[ \delta_{H}(u,v) = \frac{d_{H}(u,c) + d_{H}(c,v)}{d_S(u,v)} = \frac{w_{u,c} + w_{c,v}}{d_S(u,v)} \]
\noindent where $w_{x,y}$ is the weight of the edge connecting $x$ and $y$ in $H$.  Hence the dilation of $H$ may be computed by
\[ \Delta_{H} = \max_{u,v \in H} \frac{w_{u,c} + w_{c,v}}{d_S(u,c)} .\]
\noindent This equation lays the foundation for our formulation of the minimum dilation star problem as a linear program.

\begin{definition}
\label{definition:L}
Let $\mathcal{L}$ be the following linear program, defined over the variables $\lambda$ and $c_v$ for every $v \in S$:

\begin{center}
\noindent Minimize $\lambda$
\end{center}
\noindent such that for any $v \in S$,
\begin{equation}
\label{equation:distance_nonzero}
c_v \geq 0 ,
\end{equation}
\noindent and for any $v,w \in S$,
\begin{eqnarray}
\label{equation:distance_triangle}c_v + c_w & \geq & d_S(v,w) \\ 
\label{equation:distance_dilation}c_v + c_w & \leq & \lambda \cdot d_S(v,w) .
\end{eqnarray}
\noindent Let $\lambda^*$ be the value assigned to $\lambda$ in the optimal solution to $\mathcal{L}$.  In other words, $\lambda^*$ is the smallest dilation admitted by any set of distances satisfying all the constraints of $\mathcal{L}$.
\end{definition}

\noindent $\mathcal{L}$ is clearly feasible.  For example, if $D=\max_{x,y \in S} d_S(x,y)$, then the solution $\forall v \enspace c_v=D$ and $\lambda=2D/\min_{x,y \in S} d_S(x,y)$ is a feasible, though poor, solution.

\begin{lemma}
For any optimal solution of $\mathcal{L}$, the value of $\lambda$ gives the minimum dilation of any star network spanning $S$, and the $c_v$ values give the edge lengths of an optimal star network spanning $S$.
\end{lemma}

\begin{proof}
Each variable $c_v$ corresponds to the weight $w_{v,c}$ of the edge between $c$ and $v$ in $H$.  Inequality \ref{equation:distance_nonzero} ensures that the distances are nonnegative, Inequality
\ref{equation:distance_triangle} ensures that they obey the triangle inequality, and Inequality
\ref{equation:distance_dilation} dictates that $\lambda$ is a largest dilation among any pair of sites from $S$.  The value of $\lambda$ is optimal since $\mathcal{L}$ is defined to minimize $\lambda$.
\end{proof}

Unfortunately $\mathcal{L}$ contains $O(n)$ variables and $O(n^2)$ constraints.  Such a program could be solved using general purpose techniques in a number of steps that is a high-order polynomial in $n$ and the number of bits of precision used, but our objective is to obtain a fast algorithm whose running time is strongly polynomial in $n$.  Megiddo showed \cite{megiddo:347} that linear programs with at most two variables per inequality may be solved in strongly polynomial time; however our type (\ref{equation:distance_dilation}) inequalities have three variables, so those results cannot be applied to our problem.

\section{Reduction to Parameteric Negative Weight Cycle Detection}
\label{section:reduction}

In this section we describe a subroutine that maps the set of sites $S$ to a directed, parametrically-weighted $\lambda$-graph $G(\lambda)$.  Every edge of $G(\lambda)$ is weighted according to a nondecreasing linear function of a single graph-global variable $\lambda$.  An important property of $G(\lambda)$ is that the set of values of $\lambda$ that cause $G(\lambda)$ to contain a negative weight cycle is identical to the set of values of $\lambda$ that cause the linear program $\mathcal{L}$ to be infeasible.  Thus any assignment of $\lambda$ for which $G(\lambda)$ contains no negative weight cycles may be used in a feasible solution to $\mathcal{L}$.

\begin{definition}
\label{definition:lambda_graph}
A \emph{$\lambda$-graph} is a connected, weighted, directed graph, where the weight $w(e)$ of any edge $e$ is defined by a linear function in the form
\[ w(e) = \lambda \cdot m_e + b_e ,\]
\noindent where $m_e$ and $b_e$ are real numbers and $m_e \geq 0$.
\end{definition}

\begin{definition}
\label{definition:G}
Let $G(\lambda)$ be the $\lambda$-graph corresponding to a particular set of input sites $S$.  $G(\lambda)$ has vertices $\overline{s}$ and $\underline{s}$ for each $s \in S$.  For $s,t \in S$, $G(\lambda)$ has an edge of length $-d_S(s,t)$ from $\underline{s}$ to $\overline{t}$, and for $s \ne t$, $G(\lambda)$ has an edge of length $\lambda \cdot d_S(s,t)$ from $\overline{s}$ to $\underline{t}$.
\end{definition}

\noindent Note that an edge from $\underline{s}$ to $\overline{t}$ has weight $-d_S(s,s)=0$ when $s=t$.  An example $\lambda$-graph $G(\lambda)$ for $n=3$ is shown in Figure~\ref{figure:G}.

\begin{figure}
\centering
\scalebox{\figurescale}{\includegraphics{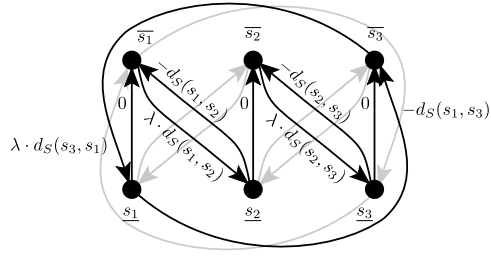}}
\caption{The graph $G(\lambda)$ for $n=3$.  The weights of grayed edges are omitted.}
\label{figure:G}
\end{figure}

\begin{lemma}
\label{lemma:make_graph_basics}
$G(\lambda)$ may be constructed in $O(n^2)$ time.
\end{lemma}

\begin{proof}
$G(\lambda)$ has $2n$ vertices and $O(n^2)$ edges, each of which may be initialized in constant time.
\end{proof}

\omit{
We now prove that feasibility of $\mathcal{L}$ is equivalent to the absence of negative weight cycles in $G(\lambda)$.  We begin by showing that feasibility implies the absence of negative weight cycles.
}

\begin{lemma}
\label{lemma:nnwc_implies_feasible}
If $\lambda \geq 1$ is assigned such that $\mathcal{L}$ has a feasible solution, then $G(\lambda)$ contains no negative weight cycle.
\end{lemma}

\begin{proof}
Since $G(\lambda)$ is bipartite, any sequence of edges $M$ traversed by a cycle in $G(\lambda)$ has even length.  Depending on which partition $M$ begins with, the sequence either takes the form
\[ M = \langle
(\overline{s_{i_1}}, \underline{s_{i_2}}),
(\underline{s_{i_2}}, \overline{s_{i_3}}),
(\overline{s_{i_3}}, \underline{s_{i_4}}),
\ldots,
(\underline{s_{i_k}}, \overline{s_{i_1}}) 
\rangle \]
\noindent or
\[ M = \langle
(\underline{s_{i_1}}, \overline{s_{i_2}}),
(\overline{s_{i_2}}, \underline{s_{i_3}}),
(\underline{s_{i_3}}, \overline{s_{i_4}}),
\ldots,
(\overline{s_{i_k}}, \underline{s_{i_1}})
\rangle \, , \]
\noindent where $s_{i_1}, s_{i_2}, \ldots, s_{i_k}$ are vertices from $G(\lambda)$.  In either case, the cycle has weight
\begin{equation}
\label{cycle_weight}
w(M) = \lambda \cdot d_S(s_{i_1}, s_{i_2}) - d_S(s_{i_2}, s_{i_3}) + \lambda \cdot d_S(s_{i_3}, s_{i_4}) - \ldots - d_S(s_{i_k}, s_{i_1}) 
\end{equation}
\noindent by the commutativity of addition.  Since $\mathcal{L}$ is feasible, there exists some set of distances $C$ satisfying the constraints of $\mathcal{L}$, i.e.
\begin{equation}
\label{ineq1}
c_x + c_y \leq \lambda \cdot d_S(x,y) \Rightarrow (c_x+c_y)/\lambda \leq d_S(x,y)
\end{equation}
\noindent and
\begin{equation}
\label{ineq2}
c_x + c_y \geq d_S(x,y) \Rightarrow -(c_x+c_y) \leq -d_S(x,y) .
\end{equation}
\noindent Substituting (\ref{ineq1}) and (\ref{ineq2}) into (\ref{cycle_weight}), we obtain
\begin{eqnarray*}
w(M) &\geq& \lambda ((c_{i_1}+c_{i_2})/\lambda) - (c_{i_2} + c_{i_3}) + \lambda((c_{i_3} + c_{i_4})) - \ldots - (c_{i_k} + c_{i_1}) \\
&\geq& (c_{i_1}+c_{i_2}) - (c_{i_2} + c_{i_3}) + (c_{i_3} + c_{i_4}) - \ldots - (c_{i_k} + c_{i_1}) \\
&\geq& c_{i_1}-c_{i_1}+c_{i_2}-c_{i_2}+\ldots+c_{i_k}-c_{i_k} \\
&\geq& 0.
\end{eqnarray*}
\end{proof}

\omit{
We now summarize this section's results as a single theorem.
}

\begin{theorem}
\label{theorem:create_graph}
Any set $S$ of $n$ sites from a metric space may be mapped to a $\lambda$-graph $G(\lambda)$ with $O(n)$ vertices, such that for any $\lambda \geq 1$, $G(\lambda)$ contains a negative weight cycle if and only if $\mathcal{L}$ is infeasible for that value of $\lambda$.  The mapping may be accomplished in $O(n^2)$ time.
\end{theorem}

\begin{proof}
By Lemma \ref{lemma:make_graph_basics}, $G(\lambda)$ may be created in $O(n^2)$ time, and by Lemma \ref{lemma:nnwc_implies_feasible}, feasibility of $\mathcal{L}$ implies an absence of negative cycles in $G(\lambda)$.  Section \ref{section:extracting_weights} describes an algorithm that, given a value $\lambda$ for which $G(\lambda)$ has no negative cycle, generates an edge length $c_v$ for every $v \in S$ that obeys the constraints of $\mathcal{L}$.  Thus, by the correctness of that algorithm, an absence of negative cycles in $G(\lambda)$ implies feasibility of $\mathcal{L}$.
\end{proof}

\section{Searching for $\lambda^*$}
\label{section:cycle_detection}

We now turn to the problem of computing the quantity $\lambda^*$.  This problem is an example of \emph{parametric negative weight cycle detection}: given a $\lambda$-graph $G(\lambda)$, find $\lambda^*$, the smallest value such that $G(\lambda^*)$ contains no cycles of negative weight.  Our algorithm functions by maintaining a range $[\lambda_1, \lambda_2]$ which is known to contain $\lambda^*$.  Initially the range is $[-\infty, +\infty]$; over $O(\log n)$ iterations, the range is narrowed until it is small enough that $\lambda^*$ may be calculated easily.  This approach is similar in spirit to Megiddo's general parametric search framework \cite{804326,322410}, which, in loose terms, searches for the solution to an optimization problem by simulating the execution of a parallel algorithm for the corresponding decision problem.

Our algorithm is presented in Listing \ref{algorithm:computing_lambda_star}.  It is an adaptation of a parallel all pairs shortest paths algorithm based on matrix squaring \cite{savage77}.  The original algorithm uses a matrix $D_i(u,v)$, which stores the weight of the shortest path from $u$ to $v$ among paths with at most $2^i$ edges.  Each $D_i(u,v)$ may be defined as the smallest sum of two cells of $D_{i-1}$, and $D_{\lceil \log_2 n \rceil}$ defines the shortest paths in the graph.  In the context of that original algorithm, edges and paths had real-number lengths, so it was sufficient to store real numbers in $D_i$.  In the context of this paper, an edge's weight is a linear function of a variable $\lambda$; hence the weight of a path is a linear function of $\lambda$.  Unfortunately the minimum-cost path between $u$ and $v$ may be different for varying values of $\lambda$, so the weight of the shortest path from $u$ to $v$ is defined by the minima of one or more linear functions of $\lambda$.  Such a lower envelope of linear functions may be represented by a piecewise linear function; hence each element of $D_i$ must store a piecewise linear function.  Without further attention the number of breakpoints in these piecewise linear functions would grow at every iteration, and eventually operating on them would dominate our algorithm's running time.  To address this, at every iteration we choose a new interval $[\lambda_1, \lambda_2]$ that contains no breakpoints, so that every $D_i$ may be compacted down to a single linear function.

\floatname{algorithm}{Listing}
\begin{algorithm}[p]
\caption{Computing the quantity $\lambda^*$.}
\label{algorithm:computing_lambda_star}
\begin{algorithmic}[1]
\STATE {\bf INPUT:} A $\lambda$-graph $G(\lambda)$ with $n$ vertices $V$.
\STATE {\bf OUTPUT:} $\lambda^*$, the smallest value of $\lambda$ such that $G(\lambda)$ has no negative-weight cycles.
\STATE Let $\lambda_1 = -\infty$ and $\lambda_2 = +\infty$.
\STATE {\bf INVARIANT:} $\lambda_1 \leq \lambda^* \leq \lambda_2$
\STATE {\bf INVARIANT:} $D_i(u,v)$ contains a linear function that represents the length of the shortest path from $u$ to $v$ among the subset of paths that use at most $2^i$ edges, as a function of $\lambda$, for any $\lambda \in [\lambda_1,\lambda_2]$
\STATE Let $D_0$ be an $n \times n$ matrix of piecewise linear functions.
\STATE Initialize $D_0(u,v) \equiv \left\{ \begin{array}{ll}
                            0 & \mbox{if $u=v$} \\
                            \lambda\cdot m_e + b_e & \mbox{if $G(\lambda)$ contains an edge $e$ from $u$ to $v$} \\
                            +\infty & \mbox{otherwise}
                            \end{array} \right .$
\FOR{ $i=1, 2, \ldots, \lceil \log_2 n \rceil$ }
  \FOR{ $u, v \in V$ }
    \STATE $D_i(u,v) \equiv \min_{w \in V} [D_{i-1}(u,w) + D_{i-1}(w,v)]$ \label{line:define_D}
  \ENDFOR

  \STATE Let $B$ be the set of breakpoints of the piecewise linear functions stored in the entries of $D_i$.

  \STATE Perform a binary search among the values in $B$, seeking an interval bounded by two consecutive breakpoints that contains $\lambda^*$. At each step, the test value of the binary search is less than $\lambda^*$ if and only if setting $\lambda$ equal to the test value causes $G(\lambda)$ to contain a negative cycle; use the Bellman--Ford shortest paths algorithm to determine whether this is the case.

  \STATE Set $\lambda_1$ and $\lambda_2$ to the endpoints of the interval found in the previous step.

  \FOR{ $u, v \in V$ }
    \STATE Replace the piecewise linear function $D_i(u,v)$ with the equivalent linear function over the range $[\lambda_1, \lambda_2]$. \label{line:simplify_D}
  \ENDFOR  
\ENDFOR
\STATE Compute $\lambda^*$, the smallest value in the range $[\lambda_1, \lambda_2]$, such that $D_k(v,v) \geq 0$ for every $v \in V$.
\STATE \textbf{Return} $\lambda^*$.
\end{algorithmic}
\end{algorithm}

\begin{lemma}
\label{lemma:lambda_star_correct}
For any $\lambda \in [\lambda_1, \lambda_2]$, the function $D_i(u,v)$ as computed in the listing evaluates to the weight of the shortest path from $u$ to $v$ among paths with at most $2^i$ edges, or $+\infty$ if no such path exists.
\end{lemma}

\begin{proof}
We argue by induction on $i$.  In the base case $i=0$, $D_i(u,v)$ must represent the weight of shortest path from $u$ to $v$ that includes up to $2^0=1$ edges.  The only such paths are trivial paths, for which $u=v$ and $D_i(u,v)=0$, and single edge paths, for which the path length equals the edge length.

For $i \geq 1$, each $D_i(u,v)$ is first defined as the lower envelope of two entries of $D_{i-1}$ in line \ref{line:define_D}, then redefined as a strictly linear function over the new smaller range $[\lambda_1, \lambda_2]$ in line \ref{line:simplify_D}, so we argue that the lemma holds after each assignment.  In the first assignment, $D_i(u,v)$ is defined to be the lower envelope of $[D_{i-1}(u,w)+D_{i-1}(w,v)]$ for all $w \in V$; in other words, every $w \in V$ is considered as a potential ``layover'' vertex, and $D_i(u,v)$ is defined as a piecewise linear function that may be defined by differing layover vertices throughout the range $[\lambda_1,\lambda_2]$.  By the inductive hypothesis, the $D_{i-1}$ values represent weights of minimum cost paths with at most $2^{i-1}$ edges; hence the resulting $D_i$ values represent weights of minimum cost paths with at most $2^{i-1}+2^{i-1}=2^i$ edges.

When $D_i(u,v)$ is reassigned in line \ref{line:simplify_D}, the range endpoints $\lambda_1$ and $\lambda_2$ have been contracted such that no entry of $D_i$ contains breakpoints in the range $[\lambda_1, \lambda_2]$.  Hence any individual $D_i(u,v)$ has no breakpoints in that range, and is replaced by a simple linear function.  This transformation preserves the condition that $D_i(u,v)$ represents the weight of the shortest path from $u$ to $v$ for any $\lambda \in [\lambda_1, \lambda_2]$.
\end{proof}

\begin{lemma}
\label{lemma:binary_search_decider}
Given two values $\lambda_1$ and $\lambda_2$ such that $\lambda_1 < \lambda_2$, it is possible to decide whether $\lambda^* < \lambda_1$, $\lambda^* > \lambda_2$, or $\lambda^* \in [\lambda_1, \lambda_2]$, in $O(n^3)$ time.
\end{lemma}

\begin{proof}
By Lemma \ref{lemma:nnwc_implies_feasible}, for any value $\lambda'$, if $G(\lambda')$ contains a negative cycle when $\lambda=\lambda'$, then $\lambda'<\lambda^*$.  So we can determine the ordering of $\lambda_1, \lambda_2$, and $\lambda^*$ using the Bellman--Ford shortest paths algorithm \cite{bellman1958,ford1962} to detect negative cycles, as follows.  First run Bellman--Ford, substituting $\lambda=\lambda_2$ to evaluate edge weights.  If we find a negative cycle, then report that $\lambda^* > \lambda_2$.  Otherwise run Bellman--Ford for $\lambda=\lambda_1$; if we find a negative cycle, then $\lambda^*$ must be in the range $[\lambda_1, \lambda_2]$.  If not, then $\lambda^* < \lambda_1$.  This decision process invokes the Bellman--Ford algorithm once or twice, and hence takes $O(n^3)$ time.
\end{proof}

\begin{lemma}
\label{lemma:lambda_star_time}
The algorithm presented in Listing \ref{algorithm:computing_lambda_star} runs in $O(n^3 \log^2 n)$ time.
\end{lemma}

\begin{proof}
Each $D_{i-1}(u,v)$ is a linear function, so each $[D_{i-1}(u,w)+D_{i-1}(w,v)]$ is a linear function as well.  $D_i(u,v)$ is defined as the lower envelope of $n$ such linear functions, which may be computed in $O(n \log n)$ time \cite{textbook}.  So each $D_i(u,v)$ may be computed is $O(n \log n)$ time, and all $O(n^2)$ iterations of the first inner for loop take $O(n^3 \log n)$ total time.  Each $D_i(u,v)$ represents the lower envelope of $O(n)$ lines, and hence has $O(n)$ breakpoints.  So the entries of $D_i$ contain a total of $O(n^3)$ breakpoints, and they may all be collected and sorted into $B$ in $O(n^3 \log n)$ time.  Once sorted, any duplicate elements may be removed from $B$ in $O(|B|)=O(n^3)$ time.

Next our algorithm searches for a new, smaller $[\lambda_1,\lambda_2]$ range that contains $\lambda^*$.  Recall that $\lambda^*$ is the value of $\lambda$ for which $G(\lambda^*)$ contains no negative weight cycle, and every entry of $D_i$ is a piecewise linear function comprised of non-decreasing linear segments; so it is sufficient to search for the segment that intersects the $\lambda=0$ line.  We find this segment using a binary search in $B$.  At every step in the search, we decide which direction to seek using the decision process described in Lemma \ref{lemma:binary_search_decider}.  Each decision takes $O(n^3)$ time, and a binary search through the $O(n^2)$ elements of $B$ makes $O(\log n)$ decisions, so the entire binary search takes $O(n^3 \log n)$ time.

Replacing an entry of $D_i$ with a (non-piecewise) linear function may be done naively in $O(n)$ time by scanning the envelope for the piece that defines the function in the range $[\lambda_1, \lambda_2]$.  So the second inner for loop takes $O(n^3)$ total time, and the outer for loop takes a total of $O(n^3 \log^2 n)$ time.

The initialization before the outer for loop takes $O(n^2)$ time.  The last step of the algorithm is to compute $\lambda^*$, the smallest value in the range $[\lambda_1, \lambda_2]$ such that $D_k(v,v) \geq 0$ for every $v \in V$.  At this point each $D_i(u,v)$ is a non-piecewise increasing linear function, so this may be done by examining each of the $n$ linear functions $D_k(v,v)$, solving for its $\lambda$-intercept, and setting $\lambda^*$ to be the largest intercept.  This entire process takes $O(n^2)$ time, so the entire algorithm takes $O(n^3 \log^2 n)$ time.
\end{proof}

\begin{theorem}
\label{theorem:lambda_star_algorithm}
The algorithm presented in Listing \ref{algorithm:computing_lambda_star} calculates $\lambda^*$ in $O(n^3 \log^2 n)$ time.
\end{theorem}

\omit{
\begin{proof}
The theorem follows from Lemmas \ref{lemma:lambda_star_correct} and \ref{lemma:lambda_star_time}.
\end{proof}
}

\section{Extracting the Edge Weights}
\label{section:extracting_weights}

Once $\lambda^*$ has been calculated, all that remains is to calculate the weight of every edge in the output star.  Our approach is to create a new graph $G'$, which is a copy of $G(\lambda)$ with the addition of a new source node $s$ with an outgoing weight 0 edge to every $\overline{v}$ (see Figure \ref{figure:G_prime}).  We then compute the single source shortest paths of $G'$ starting at $s$, and define each $c_v$ to be a function of the shortest path lengths to $\overline{v}$ and $\underline{v}$.  This process is a straightforward application of the Bellman--Ford algorithm, and hence takes $O(n^3)$ time.  The remainder of this section is dedicated to proving the correctness of this approach.

\begin{figure}
\centering
\scalebox{\figurescale}{\includegraphics{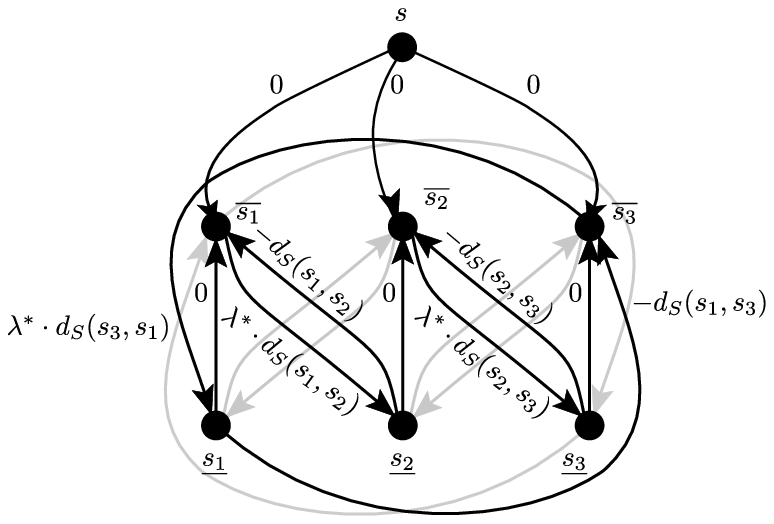}}
\caption{The graph $G'$ for $n=3$.  The weights of grayed edges are omitted.}
\label{figure:G_prime}
\end{figure}

\begin{definition}
Let $G'$ be a copy of the graph $G(\lambda)$ described in Definition~\ref{definition:G}, with all edge weights evaluated to real numbers for $\lambda=\lambda^*$, and the addition of a \emph{source vertex} $s$ with an outgoing 0-weight edge to every $\overline{v} \in G'$.  Let $P(v)$ be a shortest path from $s$ to $v$ for any vertex $v \in G'$, and let $l(v)$ be the total weight of any such $P(v)$.  The operation $P(v) \cup w$ yields the path formed by appending the edge $(v,w)$ to $P(v)$.
\end{definition}

\begin{definition}
Define $c_v = \frac{l(\underline{v}) - l(\overline{v})}{2}.$
\end{definition}

We now show that our choice of $c_v$ satisfies all three metric space properties.

\begin{lemma}
Every $c_v$ satisfies $c_v \geq 0$.
\end{lemma}

\omit{
\begin{proof}
Suppose $P(\underline{v})$ visits vertices $\langle s, \overline{x}, \ldots, \underline{v} \rangle$.  By assumption all cycles have nonnegative weight, so the cycle $\langle \overline{x}, \ldots, \underline{v}, \overline{x} \rangle$ has nonnegative weight.  Hence $l(\underline{v}) - d_S(v,x) \geq 0$.  By definition $d_S(v,x) \geq 0$, so
\[ l(\underline{v}) \geq 0 . \]
\noindent By construction $G'$ contains an edge from $s$ to $\overline{v}$ with weight 0; hence the weight of the shortest path to $\overline{v}$ must obey $l(\overline{v}) \leq 0$.  By substitution,
\begin{eqnarray*}
l(\underline{v}) & \geq & l(\overline{v}) \\
\frac{l(\underline{v}) - l(\overline{v})}{2} & \geq & 0 \\
(c_v) & \geq & 0.
\end{eqnarray*}
\end{proof}
}

\begin{proof}
For each vertex $v \in G'$ there exists an edge from $\underline{v}$ to $\overline{v}$ with weight 0.
\end{proof}

\begin{lemma}
Every distinct $c_v$ and $c_w$ satisfy $c_v + c_w \geq d_S(v,w).$
\end{lemma}
\begin{proof}
By the definition of shortest paths, we have
\begin{eqnarray*}
l(\overline{w}) &\leq& l(\underline{v}) - d_S(v,w) \\
d_S(v,w) &\leq& l(\underline{v}) - l(\overline{w}) .
\end{eqnarray*}
\noindent and by symmetric arguments,
\[ d_S(w,v) \leq l(\underline{w}) - l(\overline{v}). \]
\noindent Adding these inequalities, we obtain
\begin{eqnarray*}
d_S(v,w) + d_S(w,v) &\leq& l(\underline{v}) - l(\overline{w}) + l(\underline{w}) - l(\overline{v}) \\
d_S(v,w) &\leq& \frac{l(\underline{v}) - l(\overline{v})}{2} + \frac{l(\underline{w}) - l(\overline{w})}{2} \\ 
d_S(v,w) &\leq& (c_v) + (c_w) .
\end{eqnarray*}
\end{proof}

\begin{lemma}
Every distinct $c_v$ and $c_w$ satisfy $c_v + c_w \leq \lambda \cdot d_S(v,w).$
\end{lemma}

\begin{proof}
Observe that the path $P(\overline{w}) \cup \underline{v}$ is a path to $\underline{v}$ with weight $l(\overline{w})+\lambda \cdot d_S(w,v)$, and that the path $P(\overline{v}) \cup \underline{w}$ is a path to $\underline{w}$ with weight $l(\overline{v})+\lambda \cdot d_S(v,w)$.  By definition $P(\underline{v})$ is a shortest path to $\underline{v}$, and similarly $P(\underline{w})$ is a shortest path to $\underline{w}$, so we have
\[ l(\underline{v}) \leq l(\overline{w}) + \lambda \cdot d_S(v,w) \]
\noindent and
\[ l(\underline{w}) \leq l(\overline{v}) + \lambda \cdot d_S(v,w) .\]
\noindent Adding these inequalities, we obtain
\[ l(\underline{v}) + l(\underline{w}) \leq \left(l(\overline{w}) + \lambda \cdot d_S(w,v)\right) + \left(l(\overline{v}) + \lambda \cdot d_S(v,w)\right) .\]
\noindent By assumption $d_S(w,v)=d_S(v,w)$, so
\begin{eqnarray*}
l(\underline{v}) - l(\overline{v}) + l(\underline{w}) - l(\overline{w}) &\leq& 2 \lambda \cdot d_S(v,w) \\
\omit{\frac{l(\underline{v}) - l(\overline{v})}{2} + \frac{l(\underline{w}) - l(\overline{w})}{2} &\leq& \lambda \cdot d_S(v,w) \\}
(c_v) + (c_w) &\leq& \lambda \cdot d_S(v,w) .
\end{eqnarray*}
\end{proof}

\omit{
We now summarize this section's result in a theorem.
}

\begin{theorem}
\label{theorem:find_lambda}
Given $S$ and the corresponding $G(\lambda)$ and $\lambda^*$, a set $C$ of edge lengths $c_v$ for each $v \in S$,  such that for every $v \in S$
\[ c_v \geq 0 \]
\noindent and for every distinct $v,w \in S$
\[ c_v + c_w \geq d_S(v,w) \]
\[ c_v + c_w \leq \lambda \cdot d_S(v,w) \]
\noindent may be computed in $O(n^3)$ time.
\end{theorem}

Theorem \ref{theorem:find_lambda} establishes that for any $\lambda^*$ there exists a set $C$ of valid edge lengths.  This completes the proof of Theorem \ref{theorem:create_graph}.

\section{Conclusion}
\label{section:conclusion}

Finally we codify the main result of the paper as a theorem.

\begin{theorem}
Given a set $S \subseteq X$ of $n$ sites from a metric space $\mathcal{X} = (X,d)$, it is possible to generate a weighted star $H$ such that the distances between vertices of $H$ obey the triangle inequality, and such that $H$ has the smallest possible dilation among any such star, in $O(n^3 \log^2 n)$ time.
\end{theorem}

\omit{
\begin{proof}
The result follows from Theorems \ref{theorem:create_graph}, \ref{theorem:lambda_star_algorithm}, and \ref{theorem:find_lambda}.
\end{proof}
}

\subsubsection*{Acknowledgements}

This work was supported in part by NSF grant
0830403 and by the Office of Naval Research under grant
N00014-08-1-1015.

\small\raggedright
\bibliographystyle{abbrv} \bibliography{star_metrics}

\end{document}